\documentclass[a4paper,10pt]{article}
\usepackage{amsmath,amssymb,amsfonts,amsthm}
\usepackage{enumerate}
\usepackage[dvips]{color}
\usepackage[latin9]{inputenc}
\usepackage[T1]{fontenc}
\usepackage[english]{babel}
\usepackage{soul}

\newtheorem{theorem}{Theorem} 
\newtheorem{lemma}{Lemma}

\newtheorem{Remark}{Remark}

\newcommand{\field}[1]{\mathbb{#1}}

\newcommand{\n}{\noindent}
\newcommand{\be}{\begin{equation}}
\newcommand{\ee}{\end{equation}}
\newcommand{\ben}{\begin{equation}}

\newcommand{\een}{\end{equation}}

\newcommand{\bv}{w}
\newcommand{\bb}{\gamma}

\newcommand{\izi}{\int_{0}^{\infty}}
\newcommand{\iii}{\int_{-\infty}^{\infty}}
\newcommand{\ird}{\int_{\Rdm}}

\newcommand{\Hk}{\mathfrak{H}}
\newcommand{\Hki}{\mathfrak{H}^{-1}}

\newcommand{\Hkm}{\mathfrak{H}_\nu}
\newcommand{\Hkim}{\mathfrak{H}^{-1}_\nu}

\newcommand{\F}{\mathfrak{F}}
\newcommand{\Fi}{\mathfrak{F}^{-1}}

\newcommand{\bj}{\mathfrak{j}}

\newcommand{\Ld}{\Delta}

\newcommand{\Rdm}{\mathbb{R}^2_{+}}

\title{Linear perturbations for the vacuum axisymmetric Einstein
  equations.}
\author{Sergio Dain$^{1,2}$ and Mart\'in Reiris$^{2}$\\
  \\
  $^1$Facultad de Matem\'atica, Astronom\'{i}a y F\'{i}sica, FaMAF,\\
  Universidad Nacional de C\'ordoba,\\
  Instituto de F\'{\i}sica Enrique Gaviola, IFEG, CONICET,\\
  Ciudad Universitaria, (5000) C\'ordoba, Argentina.  \\
  $^{2}$Max Planck Institute for Gravitational Physics,\\
  (Albert Einstein Institute), Am M\"uhlenberg 1,\\
  D-14476 Potsdam Germany. }

\begin{document}
\maketitle

\begin{center}
\begin{minipage}[c]{11.5cm}
\linespread{.9}%
\selectfont
{\small 
  In axial symmetry, there is a gauge for Einstein equations such that the
  total mass of the spacetime can be written as a conserved, positive definite,
  integral on the spacelike slices. This property is expected to play an
  important role in the global evolution.  In this gauge the equations reduce
  to a coupled hyperbolic-elliptic system which is formally singular at the
  axis. Due to the rather peculiar properties of the system, the local in time existence 
  has proved to resist analysis by standard methods. To analyze the principal part of the equations, which may represent the main
  source of the difficulties, we study linear perturbation around the flat Minkowski solution
  in this gauge. In this article we solve this linearized system explicitly in terms of integral
  transformations in a remarkable simple form. This representation is well suited
  to obtain useful estimates to apply in the non-linear case.}
\vspace{1cm}

\end{minipage}
\end{center}

\section{Introduction.}
\label{sec:introduction}

The study of the gravitational systems with axi-symmetry is particularly appealing for at least a pair of important reasons. Firstly, quite a large number of interesting physical phenomena are included and can described inside this category. Most notably the Kerr family and all its derived physics and mathematics belongs to it.  Secondly because certain of its mathematical formulations enjoy a interesting number of remarkable mathematical properties. In particular there is a gauge, called the {\it maximal-isothermal gauge} in the system reduced by the axisymmetric Killing field where rather important properties, as the positivity of mass, are explicitly manifest and, presumably, would became important in the mathematical investigations of these axisymmetric systems (see \cite{Dain:2008xr}, \cite{Dain:2007pk}, \cite{Dain06c} and references therein).  

To take advantages of all this one needs to show, naturally, that the reduced Einstein equations
in the maximal-isothermal gauge is a mathematically well posed theory. As it usually happens in coordinates systems adapted to
axial symmetry, the equations are formally singular at the axis. It is the particular combination of such formally singular terms
inside the whole system of equations what makes it  difficult to treat. Until now
no well posedness result has been given in the literature. 

The analysis of the linearized equations inherit similar difficulties and therefore has proved to be non-trivial 
(see the discussion in \cite{dain10}, where this
system was analyzed numerically). In this article we solve precisely this linear problem. The remarkable algebraic properties that this system has, allows us to solve it by a combination of integral transformations suitable adapted to these equations. This solution is the
perfect analog to the solution in terms of Fourier transform of a constant
coefficient equation. The construction appears to be finely adapted to this
particular kind of singular hyperbolic-elliptic systems. 

We expect that this representation will be useful in the future to
obtain relevant estimates for the non-linear case. In this sense, we believe
that the present result opens the possibility to analyze the full axially
symmetric Einstein equations in the maximal-isothermal gauge.

The plan of the article is the following. In section \ref{sec:main-result} we
present our main result. At the end of this section we present the strategy of
the proof, which is split in three main parts discussed in sections
\ref{sec:heuristic}, \ref{sec:fourier-transform-z} and
\ref{sec:hankel-transform-rho} respectively. Finally, we include an appendix in
which we collect some properties of Bessel functions used in this article.

\subsection{Statement of the main result.}
\label{sec:main-result}  

In the maximal-isothermal gauge, the linearized Einstein equations with respect
to a flat background (in the twist-free case) reduce to the following system of
equations for the functions $v$ and $\beta$
\begin{align}
 \label{eq:119b}
  \ddot v  &= \Ld v  - \frac{\partial_\rho v }{\rho}+
  \rho \partial_\rho \left(\frac{\beta}{\rho}\right),\\
  \label{eq:120b}
  \Ld \beta & = \frac{2}{\rho} \left(\Ld v - \frac{\partial_\rho
      v}{\rho}\right).  
\end{align}
See \cite{dain10} for the deduction and physical discussion of these equations.\footnote{
We have slightly changed the notation with respect to this reference where the
function $\beta$ was denoted by $\beta^\rho$ to point out that it is the $\rho$
component of the shift vector. Since in this article we will not discuss the
other components of the linear perturbation (which can all be written in terms
of $v$ and $\beta$) to simplify the notation we drop the index $\rho$.}

In these equations the coordinates are $(t,\rho,z)$. The relevant domain is the
half plane $0\leq \rho$, $-\infty <z<\infty$, which is denoted by $\Rdm$. The
axis is given by $\rho=0$ and it defines the boundary of the domain $\Rdm$.
A dot denotes time derivative,  $\Ld$ is the flat Laplacian in 2-dimensions
\begin{equation}
  \label{eq:15b}
\Ld v = \partial^2_\rho v +  \partial^2_z  v,
\end{equation}
and $\partial$ denotes partial derivative with respect to the spatial
coordinates $\rho$ and $z$. 

The solutions $(v(t,\rho,z),\beta(t,\rho,z))$ we seek for in this article are $C^{\infty}$ functions of $\field{R}\times \Rdm$ (meaning there is a $C^{\infty}$ extension of $(v,\beta)$ into a open neighborhood of $\field{R}\times \Rdm$ in $\field{R}\times \field{R}^{2}$). For linearized gravity it is necessary to impose the following boundary and asymptotic conditions (see \cite{dain10}).

For $\beta$ it is required
\begin{enumerate}
\item $\beta (t,0,z) =0$ and for every fixed $(t,z)$, $\beta(t,\rho,z)$ is an odd function of $\rho$, 
\item For every fixed time $(t)$, $\beta =O(r^{-1})$, and $\partial ^{k} \beta=O(r^{-1-k})$, where $r=\sqrt{\rho^2+z^2}$.
\end{enumerate}

For $v$ it is required
\begin{enumerate}
\item $v(t,0,z)=(\partial_{\rho} v)(t,0,z)=0$ and for every fixed $(t,z)$, $v(t,\rho,z)$ is an even function of $\rho$.
\item For every fixed time $(t)$, $v=O(r^{-2})$ and $\partial^{k} v=O(r^{-2-k})$.
\end{enumerate}

Equation \eqref{eq:119b} is a wave equation for $v$ and so it is neecesary to
prescribe as initial data, roughly speaking, the {\it position} and {\it velocity} at the initial time $(t=0)$ which we will denote as $v|_{t=0}$ and $\dot v|_{t=0}$.

\begin{Remark} From the series expansion argument presented in \cite{dain10} it follows
that, for analytic solutions, the requirement that $v$ is even (in $\rho$) and $\beta$ odd (for all times), follows only from the condition
$\beta(t,0,z)=0$ and $v(0,0,z)=(\partial_{\rho}v)(0,0,z)=0$.  
\end{Remark}

The strategy to solve the system \eqref{eq:119b}--\eqref{eq:120b} is to use an appropriate integral
transformation to the whole set of equations to obtain simpler ones in the
transformed variables.  This integral transform is a 
combination of a Fourier transformation in the $z$ coordinate and a Hankel
transform in the $\rho$ coordinate (see the appendix \ref{sec:bess-funct-hank}
for a definition of the Hankel transform). The explicit form of the integral
transform and its inverse are given by 
\begin{align}
  \label{eq:31}
 w(t,\rho,z) & =\ird \hat w (t,k,\lambda) \,
(k|\lambda|\rho)^{\frac{1}{2}}  J_{1}(k|\lambda|\rho) e^{2\pi i\lambda z}\,d k
d\lambda,\\
 \frac{\hat w(t,k,\lambda)}{|\lambda|} & =\ird  w (t,\rho,z) \,
(k|\lambda|\rho)^{\frac{1}{2}}  J_{1}(k|\lambda|\rho) e^{-2\pi i\lambda z}\,d \rho
dz,
\end{align}
where $J_1$ is the Bessel function of the first kind of order one. The ranges for the variables are $0\leq
k<\infty$ and $-\infty < \lambda < \infty$. They define the same domain of
integration as the variables $(\rho, z)$ and hence we denote it by
the same symbol $\Rdm$. By analogy to standard physical terminology we will call
the space comprised by $(\rho,z)$ the {\it physical space}, while the one comprised by $(k|\lambda|,\lambda)$
will be called the {\it momentum space}.

\begin{theorem} {\bf (Representation of solutions.)}
\label{t:main}
  Let $F(k,\lambda)$ and $G(k,\lambda)$ be two arbitrary smooth functions of compact
  support in $\Rdm$. Define $\hat w$ and $\hat \gamma $ by
\begin{multline}
  \label{eq:33}
  \hat w (t,k,\lambda) = \frac{k^{\frac{3}{2}}}{(1+k^{2})^{\frac{1}{2}}}\int_{k}^{\infty} \left
    [ F(\bar k,\lambda)\cos\left(\lambda t(1+\bar k^{2})^{\frac{1}{2}}\right)+\right.\\ +
   \left. \frac{G(\bar k,\lambda)}{\lambda (1+\bar k^{2})^{\frac{1}{2}}} \sin\left(\lambda
      t(1+\bar k^{2})^{\frac{1}{2}}\right)\right ] \, d\bar{k},
\end{multline}
and
\begin{equation}
\hat \gamma(t,k,\lambda)  =  2\hat w (t,k,\lambda) +
\frac{2k^{\frac{3}{2}}}{(1+k^2)^{\frac{3}{2}}}\int_k^\infty \frac{(1+\bar 
      k^2)^{\frac{1}{2}}}{\bar k^{\frac{1}{2}}} \hat w (t,\bar k,\lambda)\,
    d\bar k .   \label{eq:34} 
\end{equation}
For $\hat w$ and $\hat \gamma$ define $w$ and $\gamma$ by the integral
transformation \eqref{eq:31}. Then, the functions $v=\rho^{\frac{1}{2}} \bv$
and $\beta=\rho^{-\frac{1}{2}} \bb$ define a solution of equations
\eqref{eq:119b}--\eqref{eq:120b}, whose initial data $v|_{t=0}$ and $\dot{v}|_{t=0}$
is given from $F$ and $G$ by
\begin{align}
  \label{eq:4}
  v|_{t=0} &=\sqrt{\rho} \int_{\Rdm}
  \frac{k^{\frac{3}{2}}}{(1+k^2)^{\frac{1}{2}}}\left(\int_k^\infty F(\bar k,
  \lambda)\, d\bar k \right)(k|\lambda|\rho)^{\frac{1}{2}}  J_{1}(k|\lambda|\rho)
e^{2\pi i\lambda z}\,d k d\lambda, \\
\dot{v}|_{t=0} &=\sqrt{\rho} \int_{\Rdm}
  \frac{k^{\frac{3}{2}}}{(1+k^2)^{\frac{1}{2}}}\left(\int_k^\infty G(\bar k,
  \lambda)\, d\bar k \right)(k|\lambda\rho|)^{\frac{1}{2}}  J_{1}(k|\lambda|\rho)
e^{2\pi i\lambda z}\,d k d\lambda. 
\end{align}
\end{theorem}

\begin{Remark} In this representation of solutions the boundary conditions required for $\beta$ and $v$ at the axis ({\it items 1} above) are automatically satisfied. The asymptotic conditions required at infinity ({\it items 2} above) are more difficult to establish, they require extensive technical analysis, and won't be included in the present article whose purpose is to introduce the generating formulas (\ref{eq:33}) and (\ref{eq:34}) in the most direct, simple and comprehensive fashion. In this sense we have preferred to use arbitrary functions $F$ and $G$ of compact support in the space $(k,\lambda)$. This family of generating functions is at the same time simple and rich, and allows us to avoid technical lengthy developments in the proof of Theorem \ref{t:main}.    
\end{Remark}

The proof of theorem \ref{t:main} is divided in three steps. In the first one,
we use the scale invariance of the system to introduce new rescaled
variables. The structure of the equations simplify in these new variables. This
is done in section \ref{sec:heuristic}. The second step (section
\ref{sec:fourier-transform-z}) is to use the standard Fourier transform in the
$z$ coordinate. That essentially eliminates the $z$ dependence of the
equations. Finally, in section \ref{sec:hankel-transform-rho} we analyze the
$\rho$ dependence of the equations using the Hankel transform. This is the most
important part of the article. The proof of Theorem \ref{t:main} is given thereafter.

\section{The equations in momentum space.}
\subsection{Scaling symmetry}
\label{sec:heuristic}
Equations \eqref{eq:119b}--\eqref{eq:120b} enjoy scaling symmetry (see
\cite{dain10}). This symmetry will play a fundamental role in the analysis of
these equations.  Let us describe this property.  For a given solution
$v(t,\rho,z)$ and $\beta(t,\rho,z)$ of equations
\eqref{eq:119b}--\eqref{eq:120b} we define the rescaled functions as
\begin{equation}
  \label{eq:173}
  v_s(\hat t, \hat \rho, \hat z)=v\left(\frac{t}{s},\frac{\rho}{s},\frac{z}{s}
  \right),  \quad \beta_s=
  \frac{1}{s}\beta\left(\frac{t}{s},\frac{\rho}{s},\frac{z}{s} 
  \right), 
\end{equation}
where
\begin{equation}
  \label{eq:174}
  \hat t =\frac{t}{s} , \quad \hat \rho=  \frac{\rho}{s}, \quad \hat z  =
  \frac{z}{s}. 
\end{equation}
and  $s$ is a positive real number.
Then, $v_s$ and $\beta_s$ define also a solution of equations
\eqref{eq:119b}--\eqref{eq:120b} written in terms of the rescaled coordinates $(\hat t,
\hat \rho, \hat z)$. 

Note that $v$ and $\beta$ have different scaling. This difference manifests
also in the power expansion series of these functions.

To analyze the equations it is very convenient to introduce rescaled functions
in such a way that both unknowns have the same scaling and the same 
behavior near the axis.  Let us consider
$\bv$ and $\bb$ defined by
\begin{equation}
  \label{eq:35}
\bv=\frac{v}{\sqrt{\rho}},  \quad \bb =\sqrt{\rho}\beta.  
\end{equation}
In terms of these variables, equations \eqref{eq:119b}--\eqref{eq:120b} are
given by
\begin{align}
 \label{eq:119bar}
  \ddot \bv  &= \Ld \bv  -\frac{3}{4}\frac{\bv}{\rho^2}+
  \sqrt{\rho}\partial_\rho \left(\bb \rho^{-3/2}\right),\\
  \label{eq:120bar}
  \Ld \bb-\frac{\partial_\rho \bb}{\rho}+\frac{3}{4}\frac{\bb}{\rho^2} & = 2
  \left(  \Ld \bv  -\frac{3}{4}\frac{\bv}{\rho^2}  \right) .   
\end{align}
The functions $(\bv, \bb)$ are scale invariant in the following sense. For  
a given solution  $(\bv, \bb)$ of equations
\eqref{eq:119bar}--\eqref{eq:120bar} the rescaled functions
\begin{equation}
  \label{eq:21}
  \bv_s(\hat t, \hat \rho, \hat z)=\bv\left(\frac{t}{s},\frac{\rho}{s},\frac{z}{s}
  \right),\quad  \bb_s(\hat t, \hat \rho, \hat
  z)=\bb\left(\frac{t}{s},\frac{\rho}{s},\frac{z}{s} 
  \right), 
\end{equation}
define also a solution of these equations in terms of the rescaled coordinates
\eqref{eq:174}. Also, for a given solution, we expect $\bv$ and $\bb$ to have the same
behavior at the axis, namely $\bv=\bb=O(\rho^{3/2})$. 

\begin{Remark} The important part of the rescaling \eqref{eq:35} is the relative power of
$\rho$ between $v$ and $\beta$ which compensates the different scale behavior.
We could consider scalings of the form $v=\rho^{\alpha}\bv$,
$\beta=\rho^{\alpha-1}\bb$ for any arbitrary real number $\alpha$. The specific
power chosen in \eqref{eq:35} is motivated by the Hankel transform (see
equations \eqref{eq:9}--\eqref{eq:9b} in section
\ref{sec:hankel-transform-rho}). This choice makes the appropriate integral
transformation symmetric with respect to its inverse.  Another possible choice
is $\alpha=2$. This power has the advantage that both functions $\bv$ and $\bb$
are even in $\rho$. However, the related Hankel transform is not symmetric with
respect to its inverse and the formula does not coincide with the definition
used in the literature. Hence to apply standard results concerning the
Hankel transform we need to translate them into this new definition. This makes
the proofs laborious, but there is no essential difficulty.
\end{Remark}  

The differential operators in the spatial coordinates involved in
this equation are given by 
\begin{equation}
  \label{eq:42}
  \mathbf{ P}(\bv)= \Ld \bv  -\frac{3}{4}\frac{\bv}{\rho^2},
\end{equation}
and
\begin{equation}
  \label{eq:43}
  T(\bb)=  \sqrt{\rho}\partial_\rho \left(\bb \rho^{-3/2}\right).
\end{equation}
The distinction in the notation between $\mathbf{P}$ and $T$ (boldface for
$\mathbf{P}$) is to emphasize that the differential operator $\mathbf{P}$ acts
in both coordinates $\rho$ and $z$ while $T$ does only in the $\rho$ coordinate. Later
on, we will define the operator $P$ as the $\rho$ part of $\mathbf{P}$ (see equation
 \eqref{eq:12}). As we will see, all the important features of the
equations are contained in the $\rho$ dependence. 

A remarkable fact is that the operators \eqref{eq:42} and \eqref{eq:43} are
also the natural operators for the second equation \eqref{eq:120bar}. Namely,
both equations  \eqref{eq:119bar}--\eqref{eq:120bar} are written in terms of
$\mathbf{ P}$ and $T$ as follows
\begin{align}
 \label{eq:119barP}
  \ddot \bv  &= \mathbf{P}(\bv) + T(\bb),\\
  \label{eq:120barP}
 \mathbf{P}(\bb) - T(\bb) & = 2\mathbf{P}(\bv).   
\end{align}
Note that this symmetry between both equations is not evident in terms of the
original variables $v$ and $\beta$. The fact that in the second equation
\eqref{eq:120barP} appears precisely this combination of $\mathbf{P}$ and $T$
will be crucial. {\it Our proofs will not work if we insert different (constant) coefficients
in front of $ \mathbf{P}$ and $T$ in these equations}.   

\begin{Remark} The operator $\textbf{P}$ is singular at the axis. However, this kind of
singular behavior is essentially the same as the one of the Laplace
operator in $n$-dimensions for axially-symmetric functions written in terms of
cylindrical coordinates (see, for example, the introduction of
\cite{weinstein53}).  A standard trick to avoid this problem is precisely to
work in a higher dimensional space in which the operator is regular.  This can
be done also in the case of the operator $\textbf{P}$. However, the operator
$T$ will not be regular in that higher dimensional space. It appears
not to be possible to write equations \eqref{eq:119barP}--\eqref{eq:120barP} as
regular equations on a single higher dimensional space.
\end{Remark}

It is the presence of the operator $T$ in these equations which makes them so
peculiar.  The operator $T$ is, outside the axis, a first order operator but at
the axis it is a second order operator (due to L'Hopital rule).  This behavior
indicates that we can not decompose \eqref{eq:119barP}--\eqref{eq:120barP} as a
principal part (containing only the second order operator $\textbf{P}$) plus
some lower order terms (containing only the operator $T$). This kind of
decomposition is essential to construct an iteration scheme in which each of
the equations is solved in alternative steps of the iteration. Outside the axis
this iteration scheme can be constructed, but it appears not to be possible to
include the axis (see the heuristic discussion in \cite{dain10}). In fact,
our analysis suggests that the system \eqref{eq:119barP}--\eqref{eq:120barP}
should be viewed as a unity that can not be further decomposed.

\subsection{The Fourier transform in $z$}
\label{sec:fourier-transform-z}

In equations \eqref{eq:119barP}--\eqref{eq:120barP} the $z$ dependence is clearly
simpler than the $\rho$ dependence. The equations are regular in $z$ and the
coefficients of the differential operators do not depend on $z$. 
Hence, in order to factor out the $z$ dependence we can use the
Fourier transform in this coordinate defined in the standard way by
\begin{align}
  \label{eq:F1}
 \F(f)=  \tilde f(\lambda) & =\iii f(z)  e^{-2\pi i\lambda z} \,dz,\\
 \Fi(\tilde f) =  f(z) & =\iii \tilde f(\lambda)  e^{2\pi i\lambda z}
 \,d\lambda \label{eq:F2}.  
\end{align}
Taking the Fourier transform in $z$ to equations
\eqref{eq:119barP}--\eqref{eq:120barP}, we  obtain the following equations for the 
 transformed functions $\tilde \bv(t,\rho,\lambda)$, $\tilde \bb(t,\rho,\lambda)$ 
\begin{align}
 \label{eq:119barrz}
 \ddot{\tilde \bv}& =P(\tilde \bv)-\lambda^2\tilde\bv +T(\tilde \bb),    \\ 
  \label{eq:120barrz}
 P(\tilde \bb)-\lambda^2\tilde\bb -T(\tilde \bb) &=2( P(\tilde
 \bv)-\lambda^2\tilde\bv ),   
\end{align}
where we have defined $P$ as the $\rho$ part of the operator $\mathbf{P}$, namely
\begin{equation}
  \label{eq:12}
  P(\tilde\bv)= \tilde \bv'' -\frac{3\tilde \bv}{4\rho^2}.
\end{equation}
A prime denotes derivative with respect to $\rho$. 

We use the scaling symmetry to reduce these equations to the case
$\lambda=1$ in the following way.  Define rescaled variables
\begin{equation}
  \label{eq:31b}
  \tilde t=t\lambda, \quad \tilde \rho =\rho |\lambda|, 
\end{equation}
then the rescaled functions $\tilde \bv_1(\tilde t, \tilde \rho)$, $\tilde \bb_1(\tilde
t, \tilde \rho)$ (no $\lambda$ dependence) satisfy the equations
 \begin{align}
   \label{eq:18}
 \ddot{\tilde \bv}_1& =P(\tilde \bv_1)-\tilde\bv_1 +T(\tilde \bb_1),  \\ 
   \label{eq:17}
   P(\tilde \bb_1)-\tilde\bb_1 -T(\tilde \bb_1) &=2( P(\tilde \bv_1)-\tilde\bv_1 ).
 \end{align}
 In these equations the derivatives are taken with respect to the rescaled
 coordinates \eqref{eq:31b}.  If we have a solution $\tilde \bv_1(\tilde t,
 \tilde \rho)$, $\tilde \bb_1(\tilde t, \tilde \rho)$ of equations
 \eqref{eq:18}--\eqref{eq:17}, the solution of the original equations
 \eqref{eq:119barrz}--\eqref{eq:120barrz} is given by
 \begin{equation}
   \label{eq:32}
 \tilde \bv(t,\rho,\lambda)= \tilde \bv_1( t\lambda, \rho|\lambda|),\quad  \tilde
 \bb(t,\rho,\lambda)= \tilde \bb_1( t\lambda, \rho|\lambda|). 
 \end{equation}

 The set of reduced equations \eqref{eq:18}--\eqref{eq:17} constitute our main
 equations. They encode all the main difficulties of the original equations.
 They will be solved in the next section.

\subsection{The Hankel transform in $\rho$} 
\label{sec:hankel-transform-rho}

To simplify the notation, let us write equations \eqref{eq:18}--\eqref{eq:17}
without the tilde and without the subscript $1$, namely
 \begin{align}
  \label{eq:18b}
 \ddot{\bv}& =P( \bv)-\bv +T( \bb),  \\ 
   \label{eq:17b}
   P(\bb)-\bb -T(\bb) &=2( P(\bv)-\bv ).
 \end{align} 
 The strategy to solve these equations is to expand the solution in terms of
 eigenfunctions of the operator $P$. That is, as a first step we look for a
 solution $\bj$ of the eigenvalue equation for $P$
 \begin{equation}
   \label{eq:11}
    P(\bj)=-k^2 \bj. 
 \end{equation}
By the same scaling argument  used in the previous section (with $\lambda$ replaced by
$k$), it is enough to consider the case $k=1$ 
\begin{equation}
   \label{eq:11b}
    P(\bj)=- \bj. 
 \end{equation}
 Set $\bj=\sqrt{\rho}J$, then equation \eqref{eq:11b} in terms of $J$ is given by
 \begin{equation}
   \label{eq:36}
 \rho^{2}J''+\rho J'+(\rho^{2}-1)J=0.   
 \end{equation}
 This is the Bessel equation (see equation \eqref{eq:24} in the appendix).  The
 behavior at the axis fixes the solution $J$ to be, up to a factor, the Bessel
 function of the first kind of order one, denoted by $J_{1}(\rho)$. Thus,
 $\bj=c\sqrt{\rho} J_1$ we take $c=\sqrt{k}$. After rescaling, we have that the
 eigenfunctions of \eqref{eq:11} are given by
 \begin{equation}
   \label{eq:6}
   \bj= \sqrt{k\rho}J_{1}(k\rho).
 \end{equation} 
 The solutions of equations \eqref{eq:18b}--\eqref{eq:17b} will be constructed
 as a linear superposition of the eigenfunctions $\bj$. As in the case of the
 Fourier transform with respect to the plane waves $ e^{2\pi i\lambda z}$, the
 superposition of the eigenfunctions \eqref{eq:6} lead to the following
 integral transforms
 \begin{align}
   \label{eq:9}
   \Hk(f)=  \hat f(k) & =\int_0^\infty f(\rho) \sqrt{k\rho} J_1(k\rho) \,d\rho,\\
   \Hki(\hat f) = f(\rho) & =\int_0^\infty \hat f(k) \sqrt{k\rho}
   J_1(k\rho)\,dk \label{eq:9b}.
 \end{align}
 These are Hankel transformations of first order (see equations
 \eqref{eq:1}--\eqref{eq:2} and also \cite{zemanian87} for further properties
 of the Hankel transform). The orthogonality property of the Bessel function (see
 equation \eqref{eq:orth-bessel}) ensure that $\Hki \Hk(f)=f$ for $f$ in an
 appropriate functional space (see the Appendix). We will call the space of $f(\rho)$ the ``physical
 space'' and the space of $\hat f(k)$ the ``momentum space''. 

 The rescaling \eqref{eq:35} has been tailored to obtain precisely this form of
 the Hankel transform. Other rescaling will produce integral transforms which
 are not symmetric with respect to their inverse. They will have different weights
 in $\rho$.

 Let $\hat f=\Hk(f)$ and $\hat g =\Hk(g)$, then we have the following Parseval-type of identity
 \begin{equation}
   \label{eq:10}
 \izi f(\rho)g(\rho)\, d\rho=  \izi\hat f(k) \hat g(k)  \,dk.
 \end{equation}
 That is, we have the following two  identical inner products 
\begin{equation}
\label{iso} 
 <f,g>=\izi  f(\rho)g(\rho)d\rho, \quad  <\hat f,\hat g>=\izi  \hat f(k)\hat
 g(k)dk,   
\end{equation}
defined naturally in the physical and in the momentum space respectively. 

The crucial property of the Hankel transformation $\Hk$ for our purposes is
its  behavior with respect to the differential operator $P$, namely
 \begin{equation}
   \label{eq:13}
   \Hk(P(v))=-k^2\Hk(v).
 \end{equation}
 This property follows after integration by parts and use of equation \eqref{eq:11}.
 Note also that the operator $P$ is self adjoint with respect to $<\ ,\ >$,
 namely
 \begin{equation}
   \label{eq:48}
   <w,P(v)>=  <P(w),v>.
 \end{equation}

\subsection{The Hankel transform acting on the differential operators}
\label{sec:transfer-operators}
To solve equations \eqref{eq:18b}--\eqref{eq:17b} we will apply the Hankel
transform \eqref{eq:9} to obtain simpler equations in momentum space. By
equation \eqref{eq:13}, the Hankel transform $\Hk$ acts naturally on the
operator $P$. However, this is not the case for the operator $T$. This operator
is the main source of the difficulties. The next lemma characterize the action
of $\Hk$ on $T$.

\begin{lemma}
 \label{L:bb}
 Let $\hat \bb(k)=\Hk(\bb)$ be of compact support. Then, the
 following relation holds
 \begin{equation}
   \label{eq:14}
 \Hk(T(\bb))(k)= -{k}^{\frac{3}{2}} E(k),
 \end{equation}
 where we have defined
 \begin{equation}
   \label{eq:15}
  E(k)= \int_{k}^{\infty} \frac{\hat\bb(\bar k)}{\bar k^{\frac{1}{2}}}d \bar k.
 \end{equation}
 which is a function of compact support. 
\end{lemma}

Note that the function $E(k)$ satisfies the following equation 
 \begin{equation}
   \label{eq:16}
   \frac{dE(k)}{d k}=-\frac{\hat\bb (k) }{k^{\frac{1}{2}}}.
 \end{equation}
This relation will be useful later on.

 \begin{proof}

 From the definition we have that 
 \begin{equation}
   \label{eq:39}
 \Hk(T(\bb)) (k)=\izi
\sqrt{\rho}\partial_\rho\left(\bb\rho^{-3/2}  \right) (\rho k)^{\frac{1}{2}}
J_{1}(k\rho)d\rho,    
 \end{equation}
and
\begin{equation}
  \label{eq:44}
  \bb(\rho)= \izi \hat \bb(\bar k) (\rho\bar k)^{\frac{1}{2}}
J_{1}(\bar k\rho)d\bar k.  
\end{equation}
Multiplying by $\rho^{-3/2}$ and taking a $\rho$ derivative in equation
\eqref{eq:44} under the integral, we obtain
\begin{equation}
  \label{eq:46}
  \partial_\rho\left(\bb\rho^{-\frac{3}{2}}  \right)= \izi \hat \bb(\bar k)
\bar k^{\frac{1}{2}}   \partial_\rho \left( \frac{J_{1}(\bar k\rho)}{\rho}\right) d\bar k.
\end{equation}
 We use the relation \eqref{eq:28} to compute the derivative with respect to
 $\rho$ of $J_1$.  We have
 \begin{equation}
   \label{eq:37}
 \left(\frac{J_{1}(k\rho)}{\rho}\right)'=k^{2}\frac{d}{d x}
 \left(\frac{J_{1}(x)}{x}\right)=-\frac{k}{\rho}J_{2}(k\rho).    
 \end{equation}
where we have defined $x=k\rho$. Hence we obtain
\begin{equation}
  \label{eq:47}
  \partial_\rho\left(\bb\rho^{-\frac{3}{2}}  \right)=- \izi \hat \bb(\bar k)
  \bar k^{\frac{3}{2}}
 \frac{J_{2}(\bar k\rho)}{\rho} d\bar k. 
\end{equation}
Inserting  the expression \eqref{eq:47} for  $\partial_\rho\left(\bb\rho^{-3/2}
\right)$   into equation \eqref{eq:39}  we get
 \begin{equation}
  \label{eq:40}
\Hk(T(\bb)) (k)=-\izi {k}^{\frac{1}{2}}\bar k^{\frac{3}{2}} \hat \bb(\bar k)\, d\bar k \izi
 J_{2}(\bar k\rho)J_{1}(k\rho)d\rho.  
\end{equation}

From equation \eqref{eq:26}  we know that the integral
\begin{equation}
  \label{eq:41}
 \izi J_{2}(\bar k\rho)J_{1}( k \rho)d\rho,  
\end{equation}
is equal to $k/\bar k^{2}$ if ${k}<\bar k$ and it is equal to zero if ${k}>\bar
k$. Then, the conclusion of the lemma follows.
 \end{proof}

 The next step is to analyze equation \eqref{eq:17b}. In this equation there is
 no time derivative. We want to solve this equation for an arbitrary given
 function $\bv$, which is not necessarily a solution of the other equation
 \eqref{eq:18b}. In the next lemma, we construct such solution using the Hankel
 transform.

 \begin{lemma}
 \label{l:vb}
 Let $\bv$ be a given function, with $\hat\bv(k)=\Hk(\bv)$ of compact support. Then, the solution
 $\hat\bb=\Hk(\bb)$ of compact support of the equation \eqref{eq:17b} in momentum space is given by
  \begin{equation}
    \label{eq:19}
    \hat \bb(k)= 2\hat\bv(k) +
    \frac{2k^{\frac{3}{2}}}{(1+k^2)^{\frac{3}{2}}}\int_k^\infty \frac{(1+\bar 
      k^2)^{\frac{1}{2}}}{\bar k^{\frac{1}{2}}} \hat \bv(\bar k)\, d\bar k. 
  \end{equation}
 We also have that $E(k)$, defined by \eqref{eq:15}, is given by
 \begin{equation}
   \label{eq:20}
   E(k)=  \frac{2}{(1+k^2)^{\frac{1}{2}}}\int_k^\infty \frac{(1+\bar
      k^2)^{\frac{1}{2}}}{\bar k^{\frac{1}{2}}} \hat\bv(\bar k)\, d\bar k.
 \end{equation}

 \end{lemma}

 \begin{proof}
 We apply the transform $\Hk$ to  equation  \eqref{eq:17b}. Using \eqref{eq:13} we get
 \begin{equation}
   \label{eq:22}
   -(1+k^2)\hat \bb (k)- \Hk(T(\bb))=-2(1+k^2)\hat \bv (k).
 \end{equation}
 We can use Lemma \ref{L:bb} to express $ \Hk(T(\bb)) $ in terms of $\hat
 \bb$. However, it is convenient to express everything in terms of $E(k)$
 (defined by \eqref{eq:15}) 
 instead of $\hat \bb (k)$. Using \eqref{eq:14}, \eqref{eq:15} and \eqref{eq:16} we
 obtain
 \begin{equation}
   \label{eq:23}
   \frac{dE(k)}{dk}
   +\frac{k }{1+k^{2}} E(k) =-\frac{2}{k^{1/2}}\hat \bv (k).  
 \end{equation}

 We  multiply by the integrating factor $(1+k^{2})^{\frac{1}{2}}$ to get 
 \begin{equation}
   \label{eq:45}
 \frac{d}{d
   k}\left((1+k^{2})^{\frac{1}{2}}E\right)=-2\frac{(1+{k}^{2})^{\frac{1}{2}}}{k^{\frac{1}{2}}}
 \hat \bv (k).   
 \end{equation}
 Integrating this equation and  forgetting the integrating constant to have a solution of compact support, we obtain equation
 \eqref{eq:20}. Equation \eqref{eq:19} follows directly using 
 \eqref{eq:16}. 
 \end{proof}

 Finally, in the next lemma we solve the whole system
 \eqref{eq:18b}--\eqref{eq:17b}.
\begin{lemma}
\label{l:sist}
  Let $F(k)$ and $G(k)$ be arbitrary functions of compact
  support. Then $\bv=\Hki( \hat \bv)$ and $\bb=\Hki(\hat \bb)$ define a
  solution of the system of equations \eqref{eq:18b}--\eqref{eq:17b}, where
\begin{align}
  \label{eq:38}
  \hat \bv (t,k) & = \frac{k^{\frac{3}{2}}}{(1+k^{2})^{\frac{1}{2}}}  \int_k^{\infty}
  \left( 
F(\bar k)\cos((1+\bar k^{2})^{\frac{1}{2}}t)+G(\bar k)\sin((1+\bar
k^{2})^{\frac{1}{2}}t)\right) d\bar{k},\\
 \hat \bb(t,k) & = 2\hat\bv(t,k) + \frac{2k^{\frac{3}{2}}}{(1+k^2)^{\frac{3}{2}}}\int_k^\infty \frac{(1+\bar
      k^2)^{\frac{1}{2}}}{\bar k^{\frac{1}{2}}} \hat \bv(t,\bar k)\, d\bar k.\label{eq:38b}
\end{align}
\end{lemma} 
\begin{proof}
  First, note the equation \eqref{eq:38b} is the solution of equation
  \eqref{eq:17b} given by lemma \ref{l:vb} if we consider $\hat \bv(t,k)$ as a
  given function. The only part we have to prove is that equation \eqref{eq:38}
  is also a solution of \eqref{eq:18b} in which $\bb$ is given by
  \eqref{eq:38b}.
  
We apply the Hankel transform to  equation \eqref{eq:18b} and we obtain the
following equation in momentum space 
 \begin{equation}
   \label{eq:7}
   \ddot{\hat {\bv}}(t,k)=-(1+k^2) \hat \bv (t,k)+ \Hk(T(\bb))  .
 \end{equation}
 Using the expression for $\Hk(T(\bb))$ obtained in lemma \ref{l:vb} we obtain
 \begin{equation}
   \label{eq:8}
    \ddot {\hat \bv}(t,k) =-(1+k^2)\hat \bv (t,k)-
    \frac{2k^{\frac{3}{2}}}{(1+k^2)^{\frac{1}{2}}}\int_k^\infty \frac{(1+\bar 
    k^2)^{\frac{1}{2}}}{\bar k^{\frac{1}{2}}} \hat \bv (t,\bar k)\, d\bar k . 
 \end{equation}
This equation involves only $\hat \bv$, and hence we have reduced the system
\eqref{eq:18b}--\eqref{eq:17b} to only one equation for one unknown. But
equation \eqref{eq:8} is not a differential equation, it is an
integro-differential equation. We can get a simpler expression if we define
$a(t,k)$ by
 \begin{equation}
   \label{eq:5}
 a(t,k)=\frac{(1+k^{2})^{\frac{1}{2}}}{k^{\frac{3}{2}}} \hat \bv (t,k). 
 \end{equation}
In terms of $a(t,k)$, equation \eqref{eq:8} is written as 
 \begin{equation}
   \label{Me}
 \ddot{a}(t,k)=-(1+k^{2}) a(t,k)-2\int_{k}^{\infty}\bar{k}
 a(t,\bar{k})\, d\bar{k}.
 \end{equation}
 Although equation \eqref{Me} looks certainly simpler than equation
 \eqref{eq:5}, the essential difficulty remains the same. However, a remarkable cancellation occurs if we take a
 derivative with respect to $k$ allowing us to convert this equation into a
 pure differential equation.  Namely, let us define
\begin{equation}
  \label{eq:51}
  b(t,k)=\frac{\partial}{\partial k} a(t,k).
\end{equation}
Then, taking a derivative with respect to $k$ of equation \eqref{Me} we obtain
the following remarkable simple differential equation for $b(k)$
 \begin{equation}
   \label{eq:49}
 \ddot{b}(t,k)=-(1+k^{2})b(t,k).     
 \end{equation}
The important fact is that the term $-2ka(t,k)$ that appears when the $k$
derivative is applied to $-k^2a(t,k)$  in equation \eqref{Me} cancels out by the
derivative of the integral (the factor $2$ in front of the integral is
crucial).  And hence in the final expression only appears $b(t,k)$ and not
$a(t,k)$. 

The solution of equation \eqref{eq:49} is given by
\begin{equation}
  \label{eq:52}
  b(t,k)=-F(k)\cos((1+k^{2})^{\frac{1}{2}}t)-G(k)\sin((1+k^{2})^{\frac{1}{2}}t),
\end{equation}
for arbitrary functions $F(k)$ and $G(k)$. We have written equation
\eqref{eq:52} with a minus sign just for convenience.

The function $a(t,k)$ is calculated integrating \eqref{eq:51}, namely
\begin{equation}
  \label{eq:53}
  a(t,k)= -\int_k^\infty b(t,\bar k) \, d\bar k. 
\end{equation}
There is a constant of integration that we have set to zero in \eqref{eq:53},
otherwise the function $\hat \bv$ defined by \eqref{eq:5} will not be of compact support.  

Using \eqref{eq:52}, \eqref{eq:53} and \eqref{eq:5}, expression
\eqref{eq:38} follows.   
\end{proof}

\section{\bf Proof of Theorem \ref{t:main}.}

\begin{proof} The theorem is a straightforward
consequence of lemma \ref{l:sist}, and the scaling invariance of the
equations. Namely, consider the solution of equations
\eqref{eq:18b}--\eqref{eq:17b} founded in lemma \ref{l:sist}. To be consistent
with the notation used in section \ref{sec:fourier-transform-z}, this solution
should be denoted by  $\tilde \bv_1(t,\rho)$ and $\tilde \bb_1(t,\rho)$.  Using
the scaling \eqref{eq:32}, out of  $\tilde \bv_1(t,\rho), \tilde \bb_1(t,\rho)$   we construct
\begin{equation}
  \label{eq:50}
  \tilde \bv(t,\rho,\lambda)= \tilde \bv_1(\lambda t,|\lambda| \rho),\quad 
  \tilde\bb(t,\rho,\lambda)=   \tilde \bb_1(\lambda t,|\lambda| \rho).  
\end{equation}
These are solutions of \eqref{eq:119barrz}--\eqref{eq:120barrz}. We apply the
inverse Fourier transform \eqref{eq:F2} to  $\tilde \bv(t,\rho,\lambda)$ and
$\tilde\bb(t,\rho,\lambda)$ to obtain the desired result. In \eqref{eq:33} we
have redefined the function $G(k,\lambda)$ to make simpler the connection with
the initial data. 
\end{proof}

\section*{Acknowledgments}

Most of this work took place during the visit of M. R. to FaMAF, UNC, in 2010. He
thanks  for the hospitality and support of  this institution. Part of this work
was also done during the conference ``PDEs, relativity \& nonlinear waves'',
Granada, April 5-9, 2010. The authors would like to thank the organizers of
this conference for the invitation.   

S. D. is supported by CONICET (Argentina).  This work was supported in part by
grant PIP 6354/05 of CONICET (Argentina), grant 05/B415 Secyt-UNC (Argentina)
and the Partner Group grant of the Max Planck Institute for Gravitational
Physics, Albert-Einstein-Institute (Germany).

 \appendix
 \section{Bessel functions and Hankel transform}
\label{sec:bess-funct-hank}

 We collect in this appendix properties on Bessel functions and Hankel transform
 that we use in this article.  A general reference is the classical book
 \cite{watson66} and also the book \cite{zemanian87} for the Hankel transform.  

 We will consider Bessel functions of the first kind, denoted by $J_\nu(x)$. We
 will restrict ourselves to the case where $\nu$ is a positive integer. 

 Bessel functions are 
 solutions of the Bessel equation
 \begin{equation}
   \label{eq:24}
   \frac{d^2J_\nu}{dx^2}+\frac{1}{x}\frac{dJ_\nu}{dx}+\left(1-\frac{\nu^2}{x^2}\right)
   J_\nu=0.  
 \end{equation}
 For $\nu\geq 0$, a series expansion of these function is given by
 \begin{equation}
   \label{eq:25}
   J_\nu(x)=\left(\frac{x}{2}\right)^\nu \sum_{j=0}^\infty \frac{(-1)^j}{j!
       \Gamma(j+\nu+1) }\left(\frac{x}{2}\right)^{2j},
 \end{equation}
 where $\Gamma$ denotes the standard Gamma function. Since in our case $\nu$ is
 a positive integer we  have
 $\Gamma(j+\nu+1)=(j+\nu)!$. From \eqref{eq:25} we deduce that
 $J_\nu(x)$ is an even function of $x$ for $\nu$ even and an odd function of $x$
 for $\nu$ odd. 

 We have the following orthogonality property for Bessel functions  
 \begin{equation}
   \label{eq:orth-bessel}
 \int_{0}^{\infty} \rho J_{\nu}(k\rho)J_{\nu}(\bar k\rho)d
 \rho=\frac{1}{k}\delta (\bar k-k).  
 \end{equation}

 The following integral is used in the proof of lemma \ref{L:bb} (see
 \cite{watson66}, p. 406)
 \begin{equation}
   \label{eq:26}
   \int_0^\infty J_\nu(at)J_{\nu-1}(bt)\, dt= 
 \begin{cases} 
 b^{\nu-1}/a^\nu, \\
 1/(2b),\\
 0,
 \end{cases}
 \end{equation}
 where the first value corresponds to $b<a$, the second to $b=a$ and the third
 to $b>a$. 

 We make use also of the following relation (\cite{watson66} p. 45)
 \begin{equation}
   \label{eq:27}
   x\frac{dJ_\nu(x)}{dx} -\nu J_\nu(x)=-x J_{\nu+1}(x).
 \end{equation}
 For the case $\nu=1$ this relation can be written in the form
 \begin{equation}
   \label{eq:28}
   \frac{d}{dx} \left(\frac{J_1(x)}{x}\right)= - \frac{J_2(x)}{x}.
 \end{equation}

 The Hankel transform of order $\nu$ of the function $f$ is defined as
 \begin{equation}
   \label{eq:1}
\Hkm(f) = \int_0^\infty  f(\rho) \sqrt{k\rho}  J_\nu(k\rho) d\rho, 
 \end{equation}
 and the inverse is given by
 \begin{equation}
   \label{eq:2}
\Hkim(F) = \int_0^\infty F(k) \sqrt{k\rho}  J_\nu(k\rho)  dk.
 \end{equation}
The formula \eqref{eq:orth-bessel} guarantee that $\Hkim(\Hkm(f))=f$ with $f$ 
in an appropriate functional space. More precisely the Hankel transform (of order one) is an automorphism 
on the space ${\mathcal{H}}_{1}$ of $C^{\infty}$ functions $f:\field{R}^{+}\rightarrow \field{R}$ provided with the family of seminorms $\gamma_{m,k}$ which for each pair of non-negative integers $m$ and $k$ are defined as
\begin{equation}
\gamma_{m,k}(f)=\sup_{0<\rho<\infty}|\rho^{m}(\rho^{-1}\partial_{\rho})^{k}[\rho^{-3/2}f(\rho)]|.
\end{equation}

\n A function $f$ in ${\mathcal{H}}_{1}$ has an expansion of the form
\begin{equation}
f(\rho)=\rho^{\frac{3}{2}}(a_{0}+a_{2}\rho^{2}+\ldots+a_{2l}\rho^{2l}+R_{2l}(\rho)),
\end{equation}

\n and is of fast decay, namely, decays faster than any power of $1/\rho$ as $\rho\rightarrow \infty$ (Lemma 5.2.1 in pg 130 of \cite{zemanian87})). 


\end{document}